\documentclass[twocolumn]{autart}    

\usepackage{graphicx}          
\usepackage{algpseudocode}
\usepackage{amsmath,amssymb}
\usepackage{amsfonts}
\usepackage{xcolor}
\usepackage{bm}
\usepackage{algorithm}
\usepackage{tikz}
\usetikzlibrary{shapes,arrows,calc}
\usetikzlibrary {shapes.geometric}
\usepackage{subfigure}
\usepackage{ntheorem}
\usepackage{url}
\usepackage[normalem]{ulem}

\newtheorem{remark}{Remark}

\newtheorem{lemma}{Lemma}
\newtheorem{proposition}{Proposition}

\newtheorem{definition}{Definition}
\newtheorem*{proof}{Proof}
\newtheorem{assumption}{Assumption}

\begin{document}
\pgfdeclarelayer{background}
\pgfdeclarelayer{foreground}
\pgfsetlayers{background,main,foreground}

	\begin{frontmatter}
		
		\title{Meta-learning for model-reference data-driven control 
			\thanksref{footnoteinfo}} 
		
		\thanks[footnoteinfo]{
			This paper is partially supported by the FAIR (Future Artificial Intelligence Research) project, funded by the NextGenerationEU program within the PNRR-PE-AI scheme (M4C2, Investment 1.3, Line on Artificial Intelligence).
		}
		\author[First]{Riccardo Busetto}\ead{riccardo.busetto@polimi.it}, 
		\author[Second]{~~Valentina Breschi}, 
		\author[First]{~~Simone Formentin}
		\address[First]{Dipartimento di Elettronica, Bioingegneria e Informazione, Politecnico di Milano, Via Ponzio 34/5, Milano, Italy}
		\address[Second]{Department of Electrical Engineering, Eindhoven University of Technology, 5600 MB Eindhoven, The Netherlands}


		\begin{abstract}                
		One-shot direct model-reference control design techniques, like the Virtual Reference Feedback Tuning (VRFT) approach, offer time-saving solutions for the calibration of fixed-structure controllers for dynamic systems. Nonetheless, such methods are known to be highly sensitive to the quality of the available data, often requiring long and costly experiments to attain acceptable closed-loop performance. These features might prevent the widespread adoption of such techniques, especially in low-data regimes. In this paper, we argue that the inherent similarity of many industrially relevant systems may come at hand, offering additional information from plants that are similar (yet not equal) to the system one aims to control. Assuming that this supplementary information is available, we propose a novel, direct design approach that leverages the data from similar plants, the knowledge of controllers calibrated on them, and the corresponding closed-loop performance to enhance model-reference control design. More specifically, by constructing the new controller as a combination of the available ones, our approach exploits all the available priors following a \textit{meta-learning} philosophy, while ensuring non-decreasing performance. An extensive numerical analysis supports our claims, highlighting the effectiveness of the proposed method in achieving performance comparable to iterative approaches, while at the same time retaining the efficiency of one-shot direct data-driven methods like VRFT. 	
	\end{abstract}
		
	\end{frontmatter}
	
	\section{Introduction}
        
   Real-world control applications often require calibrating parametric controllers with the same structure for systems that are \emph{similar}, yet not identical, in nature and scope. This is the case, \textit{e.g.}, when considering a platoon of mass-produced cars that navigate in diverse environments or a motor undergoing different industrial cycles. 
   
   In this paper, we will argue that in such contexts the knowledge of the controllers already tuned for some systems within a certain set and the insights into the achievable closed-loop performance  \textit{might be valuable assets to expedite and enhance the control design process} for a new plant belonging to the same family. The enabling technology to attain this goal will be the use of some similarity measures for the considered plants. 
   
   \textbf{Related literature.} A similar argument is at the heart of the so-called \emph{meta-learning} approaches (see \cite{Vanschoren2019,rivolli2022meta} and references therein), which rely on a set of \emph{meta-data} comprising all the knowledge acquired from past experience to enhance the learning of a new model or task, by looking at its similarity with the ones already identified or performed. 
However, existing meta-learning approaches mainly focus on improving the performance of classification \cite{oreshkin2018tadam,lee2019meta}, model fitting (see \emph{e.g.,} \cite{tripuraneni21a}), reinforcement learning algorithms (\cite{finn2017model,vuorio2019multimodal,Yoon2018,Hospedales22}), or, more recently, on reducing the design effort of global optimization tools \cite{busetto2023meta,chakrabarty2022optimizing,rothfuss2023meta}. 
    

In the systems and control area, strategies exist to exploit similarities for improving model fitting following a federated learning perspective (see \emph{e.g.,} \cite{BRESCHI2020a,BRESCHI2020b,ferrarotti21a,mansour2022fedcontrol} for some examples), whereas only few control design approaches have embraced the meta-learning vision.
    For instance, the work in \cite{Xin22} deals with the classical problem of identifying a linear model for a linear, time-invariant dynamical system. This work shows that using data collected from a single auxiliary system (similar to the target one) improves the finite sample accuracy of the identified model at the price of introducing a bias dictated by the difference between the auxiliary and the target systems. Shifting from linear to non-linear model fitting, the work in \cite{park2022meta} focuses on modeling a set of systems sharing the same (unknown) dynamics while determining their (possibly different) operational contexts. By casting a bilevel optimization problem to learn these unknowns from data coming from multiple systems, the authors show in simulation that the obtained model can be successfully employed within a model predictive (MPC) scheme to control the motion of a planar fully actuated rotorcraft (PFAR), even when its operational context is time-varying. Instead, the Bayesian approach presented in \cite{arcari2023bayesian} exploits meta-learning principles to enhance the description of modeling errors for a system that has to undertake multiple tasks. In particular, the authors propose to improve the system's model by exploiting data collected when it performs different assignments, to use such a model in designing model-based controllers, and prove the effectiveness of this choice on both simulated and real hardware robotic applications by coupling their strategy with a learning-based MPC scheme. 
    
    Although leveraging the meta-learning rationale, all the existing approaches still focus on learning a model of the system rather than a controller. This transition is performed in \cite{guo2023imitation} and \cite{richards2022control}. In particular, the work in \cite{guo2023imitation} focuses on reconstructing a Linear Quadratic Gaussian (LQG) regulator from closed-loop data. By exploiting the separation principle, the authors show how to effectively leverage input/output data sequences gathered by deploying both the target regulator and other LQG controllers designed with different weighting matrices. Nonetheless, their objective is to \emph{reconstruct} (and thus imitate) an existing controller rather than calibrating one from scratch. Meanwhile, the strategy presented in \cite{richards2022control} leverages ensemble models (constructed from different input/output datasets) to retrieve a parametric adaptive controller that effectively copes with unmodelled disturbances. Through simulation examples involving the motion control of planar fully actuated and underactuated quadrotors subject to wind, the authors show the effectiveness of this control-oriented meta-learning rationale and its superiority over a model-oriented one. Nevertheless, this control design approach requires a \emph{preliminary identification phase} that can be lengthy and expensive, especially if the number of ensemble models is high. Therefore, none of these methods is model-free and tailored to design a controller directly from data. 

    \textbf{Contributions.} In this work, we target the above research gap by proposing a novel \emph{meta-design} approach for calibrating fixed-structure controllers directly from data, without requiring any preliminary identification phase. By focusing on controlling a linear time-invariant system with unknown dynamics, the stepping stone of our method is the well-known Virtual Reference Feedback Tuning (VRFT) approach \cite{campi2002,formentin2019deterministic}. As for this technique, our first contribution (\textbf{Contribution C1}) is thus to carry out a calibration procedure to tightly match the behavior of a user-defined reference model, while at the same time incorporating additional meta-data into the design process, so as to ultimately obtain what we will call the \emph{meta-controller}. 
    
    Our meta-dataset will comprise open-loop data gathered from systems similar to the one we aim to control, in addition to controllers with the chosen prefixed structure tailored for these systems and the outcomes of closed-loop tests run with them. This last assumption is deemed reasonable since it is unlikely for one to exploit an existing controller to tune a new one before testing it. In line with the meta-learning rationale, we propose to use this additional meta-data to retrieve the controller for the new system \textit{by optimizing a convex combination of the controllers available in the meta-dataset}. Our extensive analysis of the ideal setup (\emph{i.e.,} model-based and noise-free) allows us to show that this structural choice leads to \emph{non-degrading matching performance} whenever the new system belongs to the set of plants used to construct the meta-dataset (which however is our standing assumption). The latter will be the second contribution of this paper (\textbf{Contribution C2}). 
    
    Moreover, we will prove (\textbf{Contribution C3}) that the closed-loop matching error attained with the new controller is bounded, and the bound depends on $(i)$ the performance experienced in the meta-dataset and $(ii)$ the similarity of the other plants and the new controlled system. We exploit these theoretical insights to augment the VRFT loss for meta-design with two additional regularization terms (\textbf{Contribution C4}), thus shaping the calibrated convex combination based on data-driven indicators of similarity and measured closed-loop performance. 
    
    Note that our structural choice to consider the convex combination of controllers has never been considered in similar works (see, \textit{e.g.}, \cite{agarwal2020boosting,kumar2021learning} for approaches exploiting the same structural assumption), as it may lead to stability problems even when all the controllers in the meta-dataset are individually stabilizing the new system. However, by leveraging the seminal work \cite{van2011data}, we establish (\textbf{Contribution C5}) theoretical and practical sufficient conditions on the controllers within the meta-dataset, for their combination to result in a controller that stabilizes the closed loop. We will then translate the above stability constraints into their data-driven counterparts and incorporate them into the direct design procedure. 
    
    As a final contribution (\textbf{Contribution C6}), we test the proposed approach in the tuning of a PI (Proportional Integral) controller within a field-oriented scheme for brushless DC motors. The natural differences arising between multiple instances of these motors due to the manufacturing process and their operational conditions once deployed, along with the industrial relevance of containing control calibration time for each new motor instance, make this example an ideal benchmark to analyze the impact of the meta-learning approach to data-driven design. By using a relatively small dataset gathered from the motor to be controlled, our results show that the proposed meta-design strategy leads to dramatically improved closed-loop performance with respect to those achieved with a controller tuned with the classical VRFT approach. A performance enhancement is experienced with the proposed method even when compared to more robust (and demanding) iterative approaches.
	
    The paper is organized as follows. The objectives of the work are introduced in Section~\ref{sec:statement}, along with the first mathematical formulation of the meta-design problem. Its main properties are then discussed in Section~\ref{sec:properties}. Based on these features, we introduce the data-driven counterpart of the considered problem in Section~\ref{sec:data_meta}. The effectiveness of the proposed direct, meta-design strategy is finally shown in a simulation case study in Section~\ref{sec:data_meta}. The paper is ended by some concluding remarks.
    
    \textbf{Notation.} Let $\mathbb{N}$ be the set of natural numbers (including zero), $\mathbb{R}$ denote the set of real numbers and $\mathbb{R}^{n}$ indicate the set of real column vectors of dimension $n$. Given $a \in \mathbb{R}^{n}$, we denote its transpose as $a^{\top}$ and $k$-th component as $[a]_{k}$, for $k=1,\ldots,n$. Still considering this vector and an $n \times n$ real matrix $A \in \mathbb{R}^{n \times n}$, we compactly denote the quadratic form $a^{\top}Aa$ as $\|a\|_{A}^{2}$. Given a set of matrices $\{A_{k} \in \mathbb{R}^{n \times n}\}_{k=1}^{N}$, $\mathrm{diag}(A_{1},\ldots,A_{N}) \in \mathbb{R}^{nN \times Nn}$ is the block-diagonal matrix having them as its diagonal entries. Given a stochastic process $v \in \mathbb{R}$, its expected value is $\mathbb{E}[v]$ while its variance is denoted as $\mathrm{var}[v]$. The uniform distribution within the interval $[a,b]$ is denoted as $\mathcal{U}_{[a,b]}$.

	\section{Setting and goal}\label{sec:statement}
	\begin{figure}[!tb]
		\centering
		\begin{tikzpicture}
					\node[coordinate] (reference) {};
					\node[draw,circle,right of=reference,node distance=.75cm] (sum1) {};
					\node[coordinate,right of=sum1,node distance=1.25cm] (aid4) {};
					\node[draw,rectangle,right of=aid4,node distance=1cm] (Ci) {$C(\theta_{k})$};
					\node [draw,isosceles triangle,right of=Ci,node distance=1.5cm] (alphai) {$[\alpha]_{k}$};
					\node[draw,circle,right of=alphai,node distance=1.25cm] (sum2) {};
					\node[rectangle,above of=Ci,node distance=.75cm] (aid6) {$\vdots$};
					\node[rectangle,above of=alphai,node distance=.75cm] (aid6bis) {$\vdots$};
					\node[rectangle,below of=Ci,node distance=.6cm] (aid7) {$\vdots$};		
					\node[rectangle,below of=alphai,node distance=.6cm] (aid7bis) {$\vdots$};		
					\node[draw,rectangle,above of=aid6,node distance=.6cm] (C1) {$C(\theta_{1})$};
					\node [draw,isosceles triangle,right of=C1,node distance=1.5cm] (alpha1) {$\alpha_{1}$};
					\node[draw,rectangle,below of=aid7,node distance=.75cm] (CN) {$C(\theta_{N})$};
					\node [draw,isosceles triangle,right of=CN,node distance=1.5cm] (alphaN) {$\alpha_{N}$};
					\node[draw,rectangle,right of=sum2,node distance=1.5cm,minimum size=7mm] (plant) {$G$};
					\node[coordinate,right of=plant,node distance=.75cm] (aid1) {};
					\node[coordinate,right of=aid1,node distance=.5cm] (output) {};
					\node[coordinate,below of=CN,node distance=1cm] (aid2) {};
					\node[rectangle,below of=aid2,node distance=.1cm] (aid3) {};
					
					\draw[->] (reference) -- node[near start,yshift=.2cm]{$r(t)$} (sum1);
					\draw[-] (sum1) -- node[yshift=.2cm]{$e^{\mathrm{o}}(t)$}(aid4);
					\draw[->] (aid4) -- (Ci);
					\draw[->] (aid4) |- (C1);
					\draw[->] (aid4) |- (CN);
					\draw[->] (C1) -- (alpha1);
					\draw[->] (Ci) -- (alphai);
					\draw[->] (CN) -- (alphaN);
					\draw[->] (alpha1) -| (sum2);
					\draw[->] (alphai) -- (sum2);
					\draw[->] (alphaN) -| (sum2);
					\draw[->] (sum2) -- node[yshift=.2cm]{$u(t)$}(plant);
					\draw[->] (plant) -- node[yshift=.2cm,near end]{$y^{\mathrm{o}}(t)$} (output);
					\draw[-] (aid1) |- (aid2); 
					\draw[->] (aid2) -| node[yshift=2cm,xshift=.2cm]{$-$}(sum1);
					
					\begin{pgfonlayer}{background}
						\path (C1.west |- C1.north)+(-0.6,0.3) node (a) {};
						\path (alpha1.east |- alphaN.east)+(+0.6,-0.6) node (c) {};
						\path (alphaN.south -| alphaN.east)+(+1,-0.3) node (b) {\textcolor{blue!70!black}{$C(\alpha)$}};
						\path[fill=blue!5!white,rounded corners, draw=blue!70!black, dashed]
						(a) rectangle (c);           
					\end{pgfonlayer}
			\end{tikzpicture}
		\caption{Scheme of the closed-loop system with the meta-controller $C(\alpha)$. The noise-free tracking error is denoted as $e^{\mathrm{o}}(t)=r(t)-y^{\mathrm{o}}(t)$, while the dependence on the back-shift operator is omitted for the sake of readability.}\label{fig:meta_scheme}
	\end{figure}
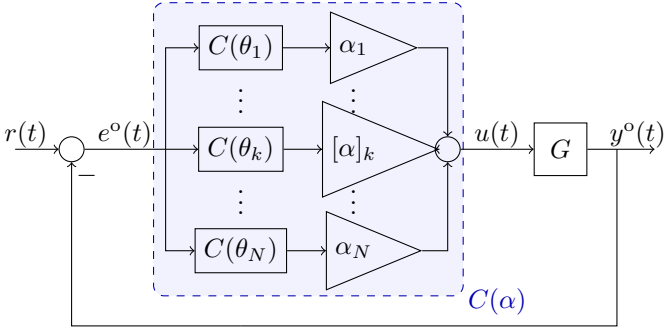

	Consider a plant $G$ within the class of systems described by the following linear, time-invariant input/output relationship:
	\begin{equation}\label{eq:system_dyn}
		\mathcal{G}:~~y^{\mathrm{o}}(t)=G(q^{-1})u(t),
	\end{equation}
	where $u(t) \in \mathbb{R}$ is the input at time $t \in \mathbb{N}$, $y^{\mathrm{o}}(t)$ is the corresponding \emph{noiseless} output, and $G(q^{-1})$ is an \emph{unknown}, proper rational function in the back-shift operator\footnote{$q^{-i}u(t)=u(t-i)$, $\forall i \in \mathbb{Z}$.} $q^{-1}$. Our aim is to design a \emph{parametric} controller for this system within the class
		\begin{equation}\label{eq:controller_class}
			\mathcal{C}(\theta):~~u(t)=\underbrace{\theta^{\top}\beta(q^{-1})}_{C(q^{-1};\theta)}(r(t)-y(t)),
		\end{equation}
	for the closed-loop system to match a desired response $y^{d}(t)$ to a reference $r(t)$ for all $t$, with $\theta \in \mathbb{R}^{n_{\theta}}$ being the set of parameters to be tuned and $\beta(q^{-1}) \in \mathbb{R}^{n_{\theta}}$ being a vector of prefixed, rational basis function in $q^{-1}$. The desired behavior is here dictated by a stable, \emph{user-defined} reference model $M$, characterized by the relationship:
	\begin{equation}\label{eq:reference_model}
		y^{d}(t)=M(q^{-1})r(t),
	\end{equation} 
	where $M(q^{-1})$ is a proper, rational function in $q^{-1}$ different from the unitary gain (\emph{i.e.,} $M(q^{-1}) \neq 1$).

	Although the true system dynamics is unknown, suppose that we have access to a finite set of input/output data\footnote{Data can be collected in open-loop when $G$ is stable, otherwise they can be gathered from closed-loop experiments, provided a stabilizing (even if poorly performing) controller is available.} $\mathcal{D}_{T}=\{u(t),y(t)\}_{t=1}^{T}$ gathered from $G$, where $y(t)$ is defined as 
	\begin{equation}\label{eq:noisy_measurements}
		y(t)=y^{\mathrm{o}}(t)+v(t),
	\end{equation}
     and $v(t)$ is the realization of a zero-mean noise at time $t$. In addition, assume that we have information on $N\geq1$ plants $\{G_{k}\}_{k=1}^{N}$, whose dynamics is also unknown but nonetheless \emph{similar} to that of $G$ according to the following definition.
     
	\begin{definition}[$\varepsilon$-similarity]\label{def:1}
		Two systems $G_{1}$ and $G_{2}$ within the class $\mathcal{G}$ in \eqref{eq:system_dyn} are said to be $\varepsilon$-similar if the following property holds:
		\begin{equation}\label{eq:varepsilon_similarity}
			\|G_{1}(q^{-1})-G_{2}(q^{-1})\|_{2}\leq \varepsilon,~~\varepsilon>0.
		\end{equation}
	\end{definition}
	Specifically, let us suppose to have access to the following data.  
        \begin{enumerate}
            \item A set of controllers $C_{k}=C(q^{-1};\theta_{k}) \subset \mathcal{C}(\theta), \ k=1,\ldots,N$, which stabilize the corresponding plants $G_{k}$ in closed-loop and have been tuned with any of the existing data-driven, model-reference strategies to match the reference model $M$ in \eqref{eq:reference_model}. 
            \item The input/output pairs $\mathcal{D}_{T}^{k}=\{u(t),y_{k}(t)\}_{t=1}^{T}$ used for such a tuning procedure, that share the same inputs of $\mathcal{D}_{T}$ while featuring outputs corrupted by a zero-mean noise $v_{k}(t)$, namely
            \begin{equation}\label{eq:noisy_metaoutputs}
                y_{k}(t)=y_{k}^{\mathrm{o}}(t)+v_{k}(t),~~~~k=1,\ldots,N.
            \end{equation}
            \item The (bounded) outputs $\mathcal{D}_{T^{\mathrm{cl}}}^{i}=\{y_{k}^{\mathrm{cl}}(t)\}_{t=1}^{T^{\mathrm{cl}}}$ obtained when tracking a prefixed (bounded) reference $\{\tilde{r}(t)\}_{t=1}^{T^{\mathrm{cl}}}$ by closing the loop on $G_{k}$ with the data-driven controller $C_{k}$, with $k=1,\ldots,N$.
        \end{enumerate}
        These elements constitute what we indicate as the \emph{meta-dataset}: 
	\begin{equation}\label{eq:meta_dataset}
		\mathcal{D}_{N}^{\mathrm{meta}}=\{C_{k},\mathcal{D}_{T}^{k},\mathcal{D}_{T^{\mathrm{cl}}}^{k}\}_{k=1}^{N},
	\end{equation}
        that, together with $\mathcal{D}_{T}$, represent the set of all the information at our disposal to design the controller for $G$, namely
        \begin{equation}\label{eq:D}
            \mathcal{D}=\{\mathcal{D}_{T} \cup \mathcal{D}_{N}^{\mathrm{meta}}\}.
        \end{equation}

	Instead of designing the new controller $C(\theta)$ for $G$ from $\mathcal{D}_{T}$ only, in this work we propose to leverage all the available information in $\mathcal{D}$ to design (in a model-reference fashion) the \emph{meta-controller} $C(\alpha) \subset \mathcal{C}(\theta)$ characterized by:
	 	\begin{subequations}\label{eq:meta_controller}
	 	\begin{equation}
	 		C(q^{-1};\alpha)\!=\!\!\sum_{k=1}^{N}[\alpha]_{k}\theta_{k}^{\top}\beta(q^{-1})\!=\!\!\sum_{k=1}^{N}[\alpha]_{k}C(q^{-1};\theta_{k}),
	 	\end{equation}
 		with
 		\begin{align}
 			& [\alpha]_{k} \geq 0,~~\forall k=1,\ldots,N,\\
 			& \sum_{k=1}^{N} [\alpha]_{k}=1,
 		\end{align}
	 \end{subequations}
 	namely the controller given by the convex combination of those available within the meta-dataset $\mathcal{D}_{N}^{\mathrm{meta}}$. 
  
    By referring to the closed-loop in \figurename{~\ref{fig:meta_scheme}}, this meta-design problem can be formalized as: 
		\begin{subequations}\label{eq:ideal_meta_problem}
			\begin{align}
				&\underset{\alpha}{\mathrm{minimize}}~~J(\alpha)\\
				& \qquad~~ \mbox{s.t. }~~C(q^{-1};\alpha)=\sum_{k=1}^{N}[\alpha]_{k}C(q^{-1};\theta_{k}),\label{eq:ideal_constr1}\\
				& \qquad~~\qquad~ [\alpha]_{k} \geq 0,~~k=1,\ldots,N,\label{eq:ideal_constr2}\\
				& \qquad~~\qquad~\sum_{k=1}^{N} [\alpha]_{k}=1,\label{eq:ideal_constr3}
			\end{align}
		where 
		\begin{equation}\label{eq:ideal_meta_loss}
			J(\alpha)\!=\!\bigg\|F(q^{-1})\!\!\left[M(q^{-1})\!-\!\frac{C(q^{-1};\alpha)G(q^{-1})}{1\!+\!C(q^{-1};\alpha)G(q^{-1})}\right]\!\!\bigg\|_{2},
		\end{equation}
		\end{subequations}
	and $F(q^{-1})$ is any user-defined weighting filter. 
	
	\section{Some properties of \eqref{eq:ideal_meta_problem} and possible variations}\label{sec:properties}
	Before shifting our attention to the data-driven counterpart of \eqref{eq:ideal_meta_problem}, let us shed light on the performance one would attain by solving this ideal meta-control problem and, consequently, introduce some changes in the original formulation to enhance the design process. To this end, we rely on the following approximation \cite[Section~2]{van2011data} of the loss in \eqref{eq:ideal_meta_problem}:
	\begin{equation}\label{eq:approx_ideal_meta_loss}
		J(\alpha) \approx \|F\Xi M-F\Xi ^{2}C(\alpha)G\|_{2},
	\end{equation} 
	where $\Xi (q^{-1})=1-M(q^{-1})$. For the sake of simplifying the notation, in the following we do not explicitly show the dependence on the backward-shift operator $q^{-1}$. 
 
	\subsection{Non-deteriorating performance}\label{sec:non_deteriorating}
	When the unknown system $G$ is not just similar but \emph{equal} to one of the plants $G_{k}$ on which $\mathcal{D}_{N}^{\mathrm{meta}}$ is build upon, with $k \in \{1,\ldots,N\}$, solving \eqref{eq:ideal_meta_problem} should either enhance the closed-loop matching performance with respect to the ones attained with $C_{k} \in \mathcal{D}_{N}^{\mathrm{meta}}$ or, in the worst case, retain the ones achieved with it. 

    Let $J_{k}$ define the loss associated with the $k$-th controller $C_{k} \in \mathcal{D}_{N}^{\mathrm{meta}}$, namely
	\begin{equation}\label{eq:J_k}
		J_{k} \!=\!\bigg\| F\!\left[M\!-\!\frac{C_{k}G}{1+C_{k}G}\right]\!\!
		\bigg\|_{2}\!\!\approx\! \|F\Xi M-F\Xi ^{2}C_{k}G\|_{2}
	\end{equation} and the solution of \eqref{eq:ideal_meta_problem} be
	\begin{equation}\label{eq:optimal_solution}
		\alpha^{\star}=\underset{\alpha \mbox{ s.t. } \eqref{eq:ideal_constr1}-\eqref{eq:ideal_constr3}}{\arg\min} J(\alpha).
	\end{equation}
	Then, we can formalize the first property of the meta-controller as follows.\\
 
	\begin{proposition}[Non-deteriorating performance]\label{prop:non_decrease}
		Let us assume that there exists one and only one $k \in \{1,\ldots,N\}$ such that the unknown plant $G$ satisfies
        \begin{equation}\label{eq:equal_plant}
		\|G-G_{k}\|_{2}=0,
	\end{equation}
        and that the controllers corresponding to the other plants are \textit{non-canceling}, namely
		\begin{equation}\label{eq:no_cancelling_controllers}
			\sum_{\substack{i=1\\i \neq k}}^{N} [\alpha]_{i}C_{i} \!=\! 0,
		\end{equation} 
        if and only if $[\alpha]_{i} \!=\! 0$, for all $i\neq k$ with $i\!=\!1,\ldots,N$.
        Then, the optimal tuning $\alpha^{\star}$ in \eqref{eq:optimal_solution} is such that 
		\begin{equation}\label{eq:non_deteriorating_ineq}
			J(\alpha^{\star}) \leq J_{k}.
		\end{equation}
	\end{proposition}
	\begin{proof}
		According to \eqref{eq:J_k}, let us rewrite $J(\alpha)$ in \eqref{eq:approx_ideal_meta_loss} as
		\begin{equation*}
			J(\alpha) \!\approx\! \|[F\Xi M\!-\!F\Xi ^{2}C_{k}G]\!+\!F\Xi ^{2}\left[C_{k}\!-\!C(\alpha)\right]G\|_{2}.
		\end{equation*}
		Based on the triangle inequality, $J(\alpha)$ satisfies
		\begin{equation}\label{eq:fundamental_inequality}
			J(\alpha) \leq J_{k}+\underbrace{\|F\Xi ^{2}\left[C_{k}\!-\!C(\alpha)\right]G\|_{2}}_{\doteq \Delta J_{k}(\alpha)},~~\forall \alpha,
		\end{equation}
		with $\Delta J_{k}(\alpha)$ being non-negative by definition. Its minimum value $\Delta J_{k}(\tilde{\alpha})=0$ is attained when 
		\begin{align*}
			&C_{k}-C(\tilde{\alpha})=0 \\
			& \Rightarrow (1-[\tilde{\alpha}]_{k})C_{k}-\sum_{\substack{i=1\\i \neq k}}^{N} [\tilde{\alpha}]_{i}C_{i}=0,
		\end{align*}
		and \eqref{eq:no_cancelling_controllers} holds for all $\alpha$. Then $\tilde{\alpha}$ must satisfy
		\begin{equation*}
			[\tilde{\alpha}]_{k}=1,~~~~[\tilde{\alpha}]_{i}=0,~i\neq k, i=1,\ldots,N,
		\end{equation*} 
		that, according to \eqref{eq:fundamental_inequality}, leads to
		\begin{equation*}
			J(\tilde{\alpha}) \leq J_{k}+\Delta J_{k}(\tilde{\alpha})=J_{k}.
		\end{equation*}	
		Since, by definition $J(\alpha^{\star}) \leq J(\bar{\alpha})$, this ends the proof. \hfill $\square$
	\end{proof}

	\subsection{Performance bounds}\label{subsec:bounds}
	Toward establishing if and when the choice of designing $C(\alpha)$ (meta-learning) instead of $C(\theta)$ (standard data-driven design) can be effective, it is fundamental to analyze the relationship between the closed-loop performance attained with the meta-controller and the actual similarity between $G$ and $\{G_{k}\}_{k=1}^{N}$, along with its link to the closed-loop performance achieved by the controllers in the meta-dataset \eqref{eq:meta_dataset}. 
	
    With this in mind, let us introduce 
	\begin{equation}\label{eq:experienced_cost}
		\tilde{J}_{k}\!=\!\bigg\| \!F\!\left[M\!-\!\frac{C_{k}G_{k}}{1+C_{k}G_{k}}\right]\!\!
		\bigg\|_2\!\approx\! \|F\Xi M-F\Xi ^{2}C_{k}G_{k}\|_2,
	\end{equation}
	indicating the performance attained by deploying the controller $C_{k}$ in feedback with the system $G_{k}$, for $k=1,\ldots,N$. Moreover, let us express the new plant $G$ as a function of the $k$-th system within the dataset, \emph{i.e.,}
	\begin{subequations}\label{eq:plant_difference}
	\begin{equation}
		G=G_{k}+\Delta G_{k},
	\end{equation}  
	where 
	\begin{equation}\label{eq:plant_difference_bound}
		\|\Delta G_{k}\|_2 \leq \varepsilon,
	\end{equation}
	according to Definition~\ref{def:1} for some (unknown) $\varepsilon$, and $\Delta G_{k}=0$ only when $G$ and $G_{k}$ are equal, for $k=1,\ldots,N$.
	\end{subequations}
	
	We can now formalize the relationship between $J(\alpha)$ in \eqref{eq:ideal_meta_loss} and $\{\tilde{J}_{k},\Delta G_{k}\}_{k=1}^{N}$ as follows. 
 
	\begin{proposition}[Bound on performance]\label{prop:bound}
		Let $G$ and $\{G_{k}\}_{k=1}^{N}$ verify \eqref{eq:plant_difference} for some (unknown) $\varepsilon$. Then, for all $\alpha$ satisfying \eqref{eq:ideal_constr2}-\eqref{eq:ideal_constr3} it holds that:  
		\begin{subequations}\label{eq:upper_bound}
			\begin{equation}\label{eq:actual_upperbound}
			J(\alpha) \leq \sum_{k=1}^{N}[\alpha]_{k} \left(\tilde{J}_{k}+\mathcal{S}_{k}\right)
			\end{equation}
			with $\{\tilde{J}_{k}\}_{k=1}^{N}$ defined as in \eqref{eq:experienced_cost} and
			\begin{equation}\label{eq:ideal_loss_similarity}
				\mathcal{S}_{k} \leq \|F\Xi ^{2}C_{k}\|_2\varepsilon, \ k=1,\ldots,N.
			\end{equation}
		\end{subequations} 
	\end{proposition}
	\begin{proof}
		Consider the approximation in \eqref{eq:approx_ideal_meta_loss} and replace $C(\alpha)$ with its definition in \eqref{eq:meta_controller}, namely
		\begin{equation}
			J(\alpha) \approx \bigg\|F \Xi M-F\Xi ^{2}\sum_{k=1}^{N}([\alpha]_{k}C_{k}G)\bigg\|_2,
		\end{equation} 
		which we can be equivalently recast  as
		\begin{equation}
			J(\alpha) \approx \bigg\|F \Xi M-F\Xi ^{2}\sum_{k=1}^{N}\{[\alpha]_{k}C_{k}(G_{k}+\Delta G_{k})\}\bigg\|_2,
		\end{equation}
		based on \eqref{eq:plant_difference}. Thanks to \eqref{eq:ideal_constr3}, it further holds that
		\begin{align*}
			& J(\alpha) \approx \bigg\|\sum_{k=1}^{N}[\alpha]_{k}F \Xi M\!-\!F\Xi ^{2}\sum_{k=1}^{N}\{[\alpha]_{k}C_{k}(G_{k}+\Delta G_{k})\}\bigg\|_2,\\
			& = \bigg\|\sum_{k=1}^{N}[\alpha]_{k}\{F \Xi M\!-\!F\Xi ^{2}C_{k}G_{k}\}\!-\!F\Xi ^{2}\sum_{k=1}^{N}[\alpha]_{k}C_{k}\Delta G_{k}\bigg\|_2.
		\end{align*}
		By using the triangle inequality twice and the Cauchy–Schwartz inequality, this approximation can be upper-bounded as follows: 
		\begin{align*}
			J(\alpha) \!&\leq \!\sum_{k=1}^{N}\! \left[\|[\alpha]_{k}\{F \Xi M\!-\!F\Xi ^{2}C_{k}G_{k}\}\|_2+\right.\\
			&\qquad \qquad \qquad \qquad \quad \left.+\|[\alpha]_{k}F\Xi ^{2}C_{k}\Delta G_{k}\|_2\right]\!\\
		&\!\!\!\!\leq\! \sum_{k=1}^{N}\! \|[\alpha]_{k}\|_2\!\left[\|F \Xi M\!-\!F\Xi ^{2}C_{k}G_{k}\|_2\!+\!\|F\Xi ^{2}C_{k}\Delta G_{k}\|_2\right]\\
			&\!\!\!=\!\sum_{k=1}^{N}[\alpha]_{k}\!\!\left[\|F \Xi M\!-\!F\Xi ^{2}C_{k}G_{k}\|_2\!+\!\|F\Xi ^{2}C_{k}\Delta G_{k}\|_2\right]\!,
		\end{align*}
		where the last equality holds thanks to \eqref{eq:ideal_constr2}. Then, the bound in \eqref{eq:actual_upperbound} straightforwardly follows from the definitions in \eqref{eq:experienced_cost} and \eqref{eq:ideal_loss_similarity}, with the upper-bound on $\mathcal{S}_{k}$ being an immediate consequence of the Cauchy–Schwartz inequality and the bound in \eqref{eq:plant_difference_bound}. \hfill $\square$
	\end{proof}

	Proposition~\ref{prop:bound} is a quantitative argument in favor of the (intuitive) fact that $\alpha$ \textquotedblleft prioritizing\textquotedblright \ controllers in $\mathcal{D}_{N}^{\mathrm{meta}}$ leading to \emph{good matching performance} and tuned on plants that are \emph{more similar} to $G$, \emph{i.e.,} such that
	\begin{equation}
		[\alpha]_{k} \gg [\alpha]_{i},~~\mbox{ if }~~\tilde{J}_{k} \ll \tilde{J}_{i}, \Delta G_{k} \ll \Delta G_{i}, 
	\end{equation}	
	for $i,k=1,\ldots,N$ and $i\neq k$, lead to lower upper bounds on the cost (and potentially lower values of $J(\alpha)$).
	  
	\subsection{Closed-loop stability}
	By solving \eqref{eq:ideal_meta_problem}, we have no guarantees that the meta-controller $C(\alpha)$ would lead to a stable closed-loop system. Inspired by \cite{van2011data}, we now point out the features that the controllers in the meta-dataset should enjoy to guarantee the stability of the meta-closed-loop, ultimately allowing us to integrate a \emph{sufficient condition} for closed-loop stability in the meta-design problem.

    To this end, let us now introduce the following quantity
	\begin{equation}\label{eq:Delta_i_def}
		\Delta_{k}=M-C_{k}G\Xi ,
	\end{equation}
	that depends on the $k$-th controller within the meta-dataset, the reference model \eqref{eq:reference_model} and the unknown plant $G$, for $k=1,\ldots,N$. By relying on this definition, consider the following property for a certain controller $C_{k}$ in $\mathcal{D}_{N}^{\mathrm{meta}}$.

	\begin{assumption}\label{assump:2}
		$C_{k} \in \mathcal{D}_{N}^{\mathrm{meta}}$ is such that:
		\begin{enumerate}
			\item[A.1] $\Delta_{k}$ in \eqref{eq:Delta_i_def} is stable;  
			\item[A.2] $\exists \delta_{k} \in (0,1)$ so that
			\begin{equation}\label{eq:condition_inftydelta}
				\|\Delta_{k}\|_{\infty} \leq \delta_{k}.
			\end{equation} 
		\end{enumerate} 
	\end{assumption}
	According to \cite[Theorem 1]{van2011data}, this implies that $C_{k}$ stabilizes the plant $G$. 
	We can now formalize the following result.
	\begin{proposition}[Meta-stability condition]
		The controller $C(\alpha)$ stabilizes the system $G$ if Assumption~\ref{assump:2} is satisfied by all $C_{k} \in \mathcal{D}_{N}^{\mathrm{meta}}$, with $k=1,\ldots,N$. 
	\end{proposition}
	\begin{proof}
		Along the line of \cite[Theorem 1]{van2011data}, $C(\alpha)$ stabilizes the system $G$ if
		\begin{enumerate}
			\item[P.1] $\Delta(\alpha)=M-C(\alpha)G\Xi $ is stable; 
			\item[P.2] $\exists \delta \in (0,1)$ such that $\delta(\alpha)=\|\Delta(\alpha)\|_{\infty}\leq \delta$.
		\end{enumerate}
		Therefore, we have to show that these sufficient conditions are verified under our assumptions. To this end, let us exploit the definition of the meta-controller \eqref{eq:meta_controller} to recast $\Delta(\alpha)$ as:
		\begin{align}\label{eq:fundamental_equality}
			\nonumber \Delta(\alpha) &=M-\!\sum_{k=1}^{N}[\alpha]_{k}C_{k}G\Xi =\sum_{k=1}^{N}[\alpha]_{k}\left(M\!-\!C_{k}G\Xi \right)\\
			&=\sum_{k=1}^{N}[\alpha]_{k}\Delta_{k},
		\end{align}
		where the last equality follows from \eqref{eq:Delta_i_def}. Since $\{\Delta_{k}\}_{k=1}^{N}$ are stable according to Assumption~\ref{assump:2} and computing their convex combination does not change their poles, then $\Delta(\alpha)$ is also stable, ultimately proving P.1. To show that P.2 holds, let us still rely on the equality in \eqref{eq:fundamental_equality}. By using the triangle and the Cauchy–Schwartz inequalities, it is straightforward to prove that
		\begin{align}\label{eq:upper_bound_proof}
			\nonumber \delta(\alpha)&=\|\Delta(\alpha)\|_{\infty}= \bigg\|\sum_{k=1}^{N}[\alpha]_{k}\Delta_{k}\bigg\|_{\infty}\\
			& \leq \sum_{k=1}^{N}\|[\alpha]_{k}\|_{\infty} \|\Delta_{k}\|_{\infty} \leq \underbrace{\sum_{k=1}^{N}\|[\alpha]_{k}\|_{\infty} \delta_{k}}_{\doteq \delta},
		\end{align}
		where the last equality holds thanks to \eqref{eq:ideal_constr2} and Assumption~\ref{assump:2}. To complete the proof we have now to show that the upper-bound $\delta$ in \eqref{eq:upper_bound_proof} lays in the interval $(0,1)$, which can be easily proven by relying once more on Assumption~\ref{assump:2}. To this end, let us define:
		\begin{equation}
			\bar{\delta}=\!\!\!\!\max_{k \in \{1,\ldots,N\}} \delta_{k},~~~~~\underline{\delta}=\!\!\!\!\min_{k \in \{1,\ldots,N\}} \delta_{k},
		\end{equation}
		with $\bar{\delta}<1$ and $\underline{\delta}>0$ since $\delta_{k} \in (0,1)$ for all $k=1,\ldots,N$ thanks to Assumption~\ref{assump:2}. Accordingly, the following holds:
		\begin{subequations}
		\begin{align}
			\delta &= \sum_{k=1}^{N}\|[\alpha]_{k}\|_{\infty} \delta_{k} \leq \bar{\delta}~ \sum_{k=1}^{N}[\alpha]_{k}=\bar{\delta}<1,\\
			\delta &= \sum_{k=1}^{N}\|[\alpha]_{k}\|_{\infty} \delta_{k} \geq \underline{\delta}~ \underbrace{\sum_{k=1}^{N}[\alpha]_{k}}_{=1}=\underline{\delta}>0,
		\end{align} 
		\end{subequations}
		also thanks to \eqref{eq:ideal_constr3}, and this concludes the proof.
		\hfill $\square$
	\end{proof}
 \begin{remark}
  Assumption~\ref{assump:2} may sound like a strong requirement as $G$ is unknown. However, when $\varepsilon$ is not too large, \textit{e.g.}, when the systems $G_{k}$ represent several instances of the same batch production, and $C_{k}$ are tuned with a suitably stability margin, it is likely that - although the performance may vary for different $k$'s - the controllers will not destabilize $G$.
  Note also that A.1 in Assumption~\ref{assump:2} can be satisfied by following the guidelines provided in \cite[Section~2]{van2011data} when tuning $\{C_{k}\}_{k=1}^{N}$.  Finally, if for any reason we doubt this assumption is not satisfied by one of the controllers in the meta-dataset, the latter can be discarded and $C(\alpha)$ can be constructed based on a reduced meta-dataset $\mathcal{D}_{N-1}^{\mathrm{meta}}$. This will be shown to be doable in a realistic setting in Remark \ref{remark4}.
 \end{remark}
 
To guarantee the verification of P.2 (see again the above proof) we can directly enforce this condition as a constraint in the tuning problem, which is thus recast as
	\begin{subequations}\label{eq:ideal_stab_meta_problem}
		\begin{align}
			&\underset{\alpha}{\mathrm{minimize}}~~J(\alpha)\\
			& \qquad~~ \mbox{s.t. }~~C(q^{-1};\alpha)=\sum_{k=1}^{N}[\alpha]_{k}C(q^{-1};\theta_{k}),\\
			& \qquad~~\qquad~ \alpha_{k} \geq 0,~~k=1,\ldots,N,\\
			& \qquad~~\qquad~\sum_{k=1}^{N} \alpha_{k}=1,\\
			& \qquad~~\qquad~\delta(\alpha) \leq \delta, \label{eq:ideal_stability_contr}
		\end{align}
	\end{subequations}
	where $\delta \in (0,1)$ becomes a tunable parameter, whose choice might lead to a (more or less) conservative tuning of $C(\alpha)$.

    	\subsection{Enhancing \eqref{eq:ideal_meta_problem} with regularization}
As highlighted by Proposition~\ref{prop:bound}, the loss $J(\alpha)$ in \eqref{eq:ideal_meta_loss} is upper-bounded by the convex combination (through $\alpha$) of the losses $\{\tilde{J}_{k}\}_{k=1}^{N}$ characterizing the matching performance of the controllers within $\mathcal{D}_{N}^{\mathrm{meta}}$, and the similarity between the new plant $G$ and $\{G_{k}\}_{k=1}^{n}$. Therefore, encouraging the choice of higher coefficients $[\alpha]_{k}$, with $k \in \{1,\ldots,N\}$, for those controllers that are characterized by small losses (and, thus, better matching performance), while being designed to control plants that are more similar to $G$, would eventually result in improved performance of the meta-controller. 
	
	To steer $\alpha$ towards a choice aligned this rationale, we propose to augment the matching cost $J(\alpha)$ in \eqref{eq:ideal_stab_meta_problem} with two regularization terms, shaping $\alpha$ based on experienced performance and plants similarities. Accordingly, the meta-design problem in \eqref{eq:ideal_stab_meta_problem} can be transformed as follows:
	\begin{subequations}\label{eq:ideal_stab_meta_problem2}
		\begin{align}
			&\underset{\alpha}{\mathrm{minimize}}~~J(\alpha)+\lambda_{J}R_{J}(\alpha;\mathbf{\tilde{J}})\!+\!\lambda_{S}R_{S}(\alpha;\mathbf{\Delta G})\label{eq:ideal_regularizedloss}\\
			& \qquad~~ \mbox{s.t. }~~C(q^{-1};\alpha)=\sum_{k=1}^{N}[\alpha]_{k}C(q^{-1};\theta_{k}),\\
			& \qquad~~\qquad~ \alpha_{k} \geq 0,~~k=1,\ldots,N,\\
			& \qquad~~\qquad~\sum_{k=1}^{N} \alpha_{k}=1,\\
			& \qquad~~\qquad~\delta(\alpha) \leq \delta,\label{eq:stab_constraint_reg}
		\end{align}
	\end{subequations}
	where $R_{J}: \mathbb{R}^{N} \rightarrow \mathbb{R}$ and $R_{S}: \mathbb{R}^{N} \rightarrow \mathbb{R}$ are two (possibly different) regularization functions, $\mathbf{\tilde{J}}=\{\tilde{J}_{k}\}_{k=1}^{N}$, $\mathbf{\Delta G}=\{\Delta G_{k}\}_{k=1}^{N}$, while $\lambda_{J},\lambda_{S}\geq0$ are tunable penalties that modulate the relative importance of the regularization terms with respect to the matching loss.     
	\section{Meta-learning in a realistic setting}\label{sec:data_meta}
    	The upgraded final meta-design problem in \eqref{eq:ideal_stab_meta_problem2} 
     (with respect to \eqref{eq:ideal_meta_problem}, it incorporate a stability constraint and suitable regularization terms) 
     is not yet applicable to a real-world problem. Indeed, it still depends on the input/output relationship of the controlled system $G$, the ideal performance $\mathbf{\tilde{J}}$ of the controllers in the meta-dataset and the mismatch $\mathbf{\Delta G}$ between $G$ and the other plants $\{G_{k}\}_{k=1}^{N}$. All the above details are not known nor directly accessible. Instead, only $\mathcal{D}$ in \eqref{eq:D} is available to solve the meta-design problem. Hence, in this section, our goal is to translate \eqref{eq:ideal_stab_meta_problem2} into its purely direct, data-driven counterpart, thus applicable in a realistic setting.  
	
	\subsection{A data-driven reformulation of $J(\alpha)$}
        Before shifting from the loss $J(\alpha)$ in \eqref{eq:ideal_meta_loss} to its data-based counterpart, let us initially replace it with its square:
        \begin{equation}\label{eq:ideal_squared_loss}
        J(\alpha)\!=\!\bigg\|F(q^{-1})\!\!\left[M(q^{-1})\!-\!\frac{C(q^{-1};\alpha)G(q^{-1})}{1\!+\!C(q^{-1};\alpha)G(q^{-1})}\right]\!\!\bigg\|_2^{2},
        \end{equation}
        which results in a design problem that is easier to handle numerically. We can now follow the footsteps of the VRFT approach \cite{campi2002} to transition toward the data-driven loss. Hence, we define the \emph{virtual reference} as the set point that would return the measured outputs if fed to the reference model, namely
    \begin{equation}\label{eq:virtual_reference}
        y(t)=M(q^{-1})r_{v}(t),~~~t=1,\ldots,T.
    \end{equation}
    in turn resulting in the virtual tracking error
    \begin{equation}
        e_{v}(t)=r_{v}(t)-y(t)=(M(q^{-1})-1)y(t) \simeq 0.
    \end{equation}
    Accordingly, the input generated by the meta-controller has to be equal to that comprised in $\mathcal{D}_{T}$. This allows us to cast the data-driven matching loss as
    \begin{equation}\label{eq:data_driven_loss}
        J^{d}(\alpha)=\sum_{t=1}^{T}\left(u^{L}(t)-C(q^{-1};\alpha)e_{v}^{L}(t)\right)^{2},
    \end{equation}
    where $u^{L}(t)$ and $e_{v}^{L}(t)$ are obtained by filtering the input and virtual error with $L(q^{-1})$ defined as 
    \begin{equation}
        L(q^{-1})=\frac{W(q^{-1})M(q^{-1})(1-M(q^{-1}))}{\Phi_{u}^{1/2}},
    \end{equation}
    with $\Phi_{u}$ being the spectral density of the input signal in $\mathcal{D}_{T}$.  
    \begin{remark}[On the choice of $L(q^{-1})$]
        Since the squared loss in \eqref{eq:ideal_squared_loss} corresponds to the standard objective of the VRFT approach for the basis
        \begin{equation}\label{eq:beta_meta}
            \beta^{\mathrm{meta}}(q^{-1})=\begin{bmatrix}
                C^{\top}(q^{-1};\theta_{1}) &
                \cdots &
                C^{\top}(q^{-1};\theta_{N})
            \end{bmatrix}^{\!\top},
        \end{equation}
        the reasoning carried out in \cite[Section~3]{campi2002} for the selection of $L(q^{-1})$ applies to the considered problem.
    \end{remark}
    \begin{remark}[Retrieving $r_{v}(t)$]
        For the virtual reference to be computed, one needs to invert the reference model. However, this inverse is non-causal every time $M(q^{-1})$ is a strictly proper rational function. This issue can be overcome by manipulating the reference model as discussed in \cite[Proposition~1 and Section~7]{FORMENTIN2016}, in turn resulting in a reduction of the samples available for design.
    \end{remark}

    While allowing us to remove the dependence on the plant model, \eqref{eq:data_driven_loss} is not yet designed to counteract the impact that noisy data have on meta-design. Once again, we overcome this issue by echoing the VRFT approach and resorting to an instrumental variable scheme \cite{soderstrom2002}. Specifically, we assume to perform an additional experiment on the plant $G$ by feeding it with the same input sequence $\{u(t)\}_{t=1}^{T}$ in $\mathcal{D}_{T}$. This allows us to gather a new set of outputs $\{y^{IV}(t)\}_{t=1}^{T}$, corrupted by measurement noise uncorrelated with the one affecting the outputs already available in $\mathcal{D}_{T}$. We can now construct the instrument as
    \begin{equation}\label{eq:instrument}
       \zeta(t)=\beta^{\mathrm{meta}}(q^{-1})e_{v}^{L}(t),
    \end{equation}
    with $\beta^{\mathrm{meta}}(q^{-1})$ defined as in \eqref{eq:beta_meta}, and recast the data-driven cost as
    \begin{equation}\label{eq:data_driven_lossIV}
         J^{d,\mathrm{IV}}(\alpha)=\sum_{t=1}^{T}\left[\zeta(t)\left(u^{L}(t)-C(q^{-1};\alpha)e_{v}^{L}(t)\right)\right]^{2}.
    \end{equation}
    \subsection{Data-driven stability constraint}\label{sec:stability_datadriven}
    To obtain the data-driven counterpart of the stability constraint in \eqref{eq:stab_constraint_reg} we rely on the strategy already proposed in \cite{van2011data}, which we recall only for the case of stable and minimum-phase $G$ due to its relevance for our numerical example\footnote{Nonetheless, the extension to unstable and non-minimum-phase plants is straightforward following \cite{van2011data} and the approach of this section.}. 
    
    Suppose that the data in $\mathcal{D}_{T}$ satisfy the assumptions in \cite[Section 4]{van2011data}. In this scenario, let us introduce
    \begin{equation}\label{eq:error_stable}
        \Delta(\alpha)u(t)\!=\!\left[Mu-C(\alpha)\Xi G\right]u(t)\!\simeq\! \underbrace{Mu(t)-C(\alpha)\Xi y(t)}_{=e_{s}(t;\alpha)}, 
    \end{equation}
    which allows us to link the quantity whose infinity norm we aim at bounding (see \eqref{eq:upper_bound_proof}) and the available data. Note that the last approximation is due to the fact that we are neglecting the impact of measurement noise, which is nonetheless legitimate when the measurement noise and the input sequence in $\mathcal{D}_{T}$ are uncorrelated (see \cite[A4, Section 4]{van2011data}). By relying on \eqref{eq:error_stable}, $\delta(\alpha)$ in \eqref{eq:ideal_stability_contr} can be approximated as
    \begin{equation}\label{eq:delta_approxStable}
        \hat{\delta}(\alpha)=\max_{\omega_{i}}\bigg|\frac{\Phi_{u,e_{s}}(\omega_{i};\alpha)}{\Phi_{u}(\omega_{i})}\bigg|
    \end{equation}
    where $\omega_{i}=2\pi i/(2\ell+1)$, for $i=0,1,\ldots,\ell+1$, and
    \begin{align*}
        & \Phi_{u,e_{s}}(\omega_{i};\alpha)=\sum_{\tau=-\ell}^{\ell}\hat{\Gamma}_{u,e_{s}}(\tau;\alpha)e^{-j\tau\omega_{i}},\\
        & \Phi_{u}(\omega_{i})= \sum_{\tau=-\ell}^{\ell}\hat{\Gamma}_{u}(\tau)e^{-j\tau\omega_{i}},
    \end{align*}
    with $j$ denoting the imaginary unit, and
    \begin{align*}
        & \hat{\Gamma}_{u,e_{s}}(\tau;\alpha)=\frac{1}{T}\sum_{t=1}^{T}u(t-\tau)e_{s}(t;\alpha),~~\tau=-\ell,\ldots,\ell,\\
        & \hat{\Gamma}_{u}(\tau)=\frac{1}{T}\sum_{t=1}^{T}u(t-\tau)u(t),~~\tau=-\ell,\ldots,\ell,
    \end{align*}
    being the sampled auto-correlation of the input and cross-correlation between the input and $e_{s}$ in \eqref{eq:error_stable}. Note that the approximation in \eqref{eq:delta_approxStable} depends on the window length $\ell$, which is an additional hyper-parameter of the design problem.   
    The reader is referred to \cite{van2011data} for additional insights on the derivation of this approximation.

    \begin{remark}[Checking $C_{k} \in \mathcal{D}_{N}^{\mathrm{meta}}$]\label{remark4}
    Apart from being used to approximate $\delta(\alpha)$ in \eqref{eq:stab_constraint_reg}, the previous strategy can be used to approximate $\|\Delta_{k}\|_{\infty}$ in \eqref{eq:condition_inftydelta}, for $k=1,\ldots,N$. In turn, this allows us to empirically evaluate if the controllers added to the meta-dataset stabilize $G$, eventually discarding them if they do not satisfy the data-driven stability condition $\widehat{\|\Delta_{k}\|}_{\infty}\leq \delta_{k}$. 
    \end{remark}
    
    \subsection{Data-driven design of the regularization penalties}
    Let us then focus on the two regularization terms in \eqref{eq:ideal_regularizedloss}. To translate them into their data-driven counterparts, we have to define two data-driven indicators for the experienced performance and the plants' similarities to replace $\mathbf{\tilde{J}}$ and $\mathbf{\Delta G}$, respectively. 

    Under the assumption that the controllers in the meta-dataset have been deployed and tested in closed-loop, we propose to substitute $\mathbf{\tilde{J}}$ with the squared (observed) closed-loop matching error
    \begin{equation}\label{eq:new_indexJtilde}
        \tilde{J}_{k} \rightarrow \tilde{J}_{k}^{d}=\sum_{t=1}^{T^{\mathrm{cl}}}(y^{d}(t)-y_{k}^{\mathrm{cl}}(t))^{2},
    \end{equation}
    where $y^{d}(t)$ is the desired output for a prefixed reference $\tilde{r}(t)$ over $T^{\mathrm{cl}}$ time steps, while $y_{k}^{\mathrm{cl}}(t)$ is the observed closed-loop output comprised in $\mathcal{D}_{N}^{\mathrm{meta}}$, for $k=1,\ldots,N$ and $t=1,\ldots,T^{\mathrm{cl}}$. Meanwhile, since we assume that both $\mathcal{D}_{T}$ and $\mathcal{D}_{T}^{k}$ comprise the same input sequence, we exploit the squared difference between the measured open-loop outputs as a proxy of $\mathbf{\Delta G}$, \emph{i.e.,}
    \begin{equation}\label{eq:new_indexS}
        \Delta G_{k} \rightarrow S_{k}=\sum_{t=1}^{T}(y(t)-y_{k}(t))^{2},~~k=1,\ldots,N.
    \end{equation}
    This choice is in line with our definition of similarity, according to which two plants are \textit{similar} if the average gain of their difference is limited at all frequencies. 
    
    Note that, both $\tilde{J}_{k}^{d}$ in \eqref{eq:new_indexJtilde} and $S_{k}$ in \eqref{eq:new_indexS} are built from noisy data, whose impact on their statistical means is formalized in the following lemmas\footnote{With a slight abuse of notation, stochastic processes will be denoted by dropping their dependence on $t$.}.
    \begin{lemma}[Mean features of $\tilde{J}_{k}^{d}$]
        Given $\tilde{J}_{k}^{d}$ in \eqref{eq:new_indexJtilde}, assume that the measurement noise $v_{k}^{\mathrm{cl}}$ acting on the $k$-th closed-loop output $y_{k}^{\mathrm{cl}}$ is zero-mean. Then:
        \begin{equation}\label{eq:expected_J}
        \frac{1}{T^{\mathrm{cl}}}\mathbb{E}\left[\tilde{J}_{k}^{d}\right] \!=\! \frac{1}{T^{\mathrm{cl}}}\sum_{t=1}^{T^{\mathrm{cl}}}\left(y^{d}(t)-y_{k}^{\mathrm{o,cl}}(t)
        \right)^{2}+var[\tilde{v}_{k}^{\mathrm{cl}}],
        \end{equation}
        with 
        \begin{equation}
            y_{k}^{\mathrm{o,cl}}(t)\!=\!\frac{C_{k}G_{k}}{1+C_{k}G_{k}}\tilde{r}(t),~~~~\tilde{v}_{k}^{\mathrm{cl}}\!=\!\frac{1}{1+C_{k}G_{k}}v_{k}^{\mathrm{cl}}(t).
        \end{equation}
    \end{lemma} 
    \begin{proof}
        By the superimposition principle, the error characterizing \eqref{eq:new_indexJtilde} can be equivalently rewritten as
        \begin{equation*}
            y^{d}(t)-y_{k}^{\mathrm{cl}}(t)=y^{d}(t)-\frac{C_{k}G_{k}}{1+C_{k}G_{k}}\tilde{r}(t)-\frac{1}{1+C_{k}G_{k}}v_{k}^{\mathrm{cl}}(t).
        \end{equation*}
        Since $C_{k}$ stabilizes $G_{k}$ in closed-loop by construction, the auto-regressive process $\tilde{v}_{k}^{\mathrm{cl}}$ resulting from this decomposition is weak-sense stationary. Therefore, its mean value is constant over time, and it is equal to zero because of our assumption on $v_{k}^{\mathrm{cl}}$. The result in \eqref{eq:expected_J} follows straightforwardly by bringing the term dependent on $y^{d}(t)-y_{k}^{\mathrm{o,cl}}(t)$ outside the expectation, since it is deterministic. \hfill $\square$
    \end{proof}
     This result allows us to highlight the connection between the selected data-driven performance index and $\tilde{J}_{k}$. In particular, the first term on the right-hand-side of \eqref{eq:expected_J} can be equivalently rewritten as
     \begin{equation*}
         \sum_{t=1}^{T^{\mathrm{cl}}}\left(y^{d}(t)-y_{k}^{\mathrm{o,cl}}(t)
        \right)^{2}=\sum_{t=1}^{T^{\mathrm{cl}}}\left[\left(M-\frac{C_{k}G_{k}}{1+C_{k}G_{k}}\right)\tilde{r}(t)
        \right]^{2}\!\!,
     \end{equation*}
     and upper-bounded as follows
     \begin{equation}
       \sum_{t=1}^{T^{\mathrm{cl}}}\!\left[\!\left(M\!-\!\frac{C_{k}G_{k}}{1+C_{k}G_{k}}\right)\!\tilde{r}(t)
        \right]^{2}\!\!\leq\left\|M\!-\!\frac{C_{k}G_{k}}{1+C_{k}G_{k}}\right\|_{2}^{\!2}\!\|\tilde{r}\|_{2}^{\!2},
     \end{equation}
    where $\tilde{r}$ compactly denotes $\{\tilde{r}(t)\}_{t=1}^{T^{\mathrm{cl}}}$, and the first element on the right-hand-side of the inequality corresponds to $\tilde{J}_{k}^{2}$ for $F(q^{-1})=1$.
    \begin{lemma}[Mean features of $S_{k}$]
        Assume that the zero-mean noise sequences acting on the outputs comprised in $\mathcal{D}_{T}$ and $\mathcal{D}_{T}^{k}$ are uncorrelated, namely 
        \begin{equation*}
        v(t_{1}) \perp v_{k}(t_{2}),~~~~\forall t_{1},t_{2}\geq 0,~~k=1,\ldots,N,
        \end{equation*}
        with $v$ and $v_{k}$ introduced in \eqref{eq:noisy_measurements} and \eqref{eq:noisy_metaoutputs}, respectively. Then:
        \begin{equation}\label{eq:normalized_mean}
            \frac{1}{T}\mathbb{E}\left[S_{k}\right] =\frac{1}{T}\sum_{t=1}^{T}(y^{\mathrm{o}}(t)-y_{k}^{\mathrm{o}}(t))^{2}+\mathrm{var}[v]+\mathrm{var}[v_{k}].        
        \end{equation}
    \end{lemma}
    \begin{proof}
        The proof results from the decomposition 
        \begin{equation*}
            y(t)-y_{k}(t)=y^{\mathrm{o}}(t)-y_{k}^{\mathrm{o}}(t)+v(t)-v_{k}(t).
        \end{equation*}
        Accordingly, the normalized expected value of $S_{k}$ becomes
        \begin{equation}\label{eq:expected_S}
            \frac{1}{T}\mathbb{E}\left[S_{k}\right]\!=\!\frac{1}{T}\sum_{t=1}^{T}(y^{\mathrm{o}}(t)-y_{k}^{\mathrm{o}}(t))^{2}\!+\!\frac{1}{T}\sum_{t=1}^{T}(v(t)\!-\!v_{k}(t))^{2}\!,
        \end{equation}
        where the first term is outside the expectation as it is deterministic, while mixed products disappear since $v$ and $v_{k}$ are both assumed to be zero mean. The result in \eqref{eq:normalized_mean} follows by further decomposing the square of the second term on the right-hand-side of the previous equality and by exploiting the lack of correlation between $v$ and $v_{k}$.    \hfill $\square$
    \end{proof}
    This lemma allows us to bridge between $\Delta G_{k}$ and the chosen data-based index. Indeed, the first term in \eqref{eq:expected_S} can be rewritten as
    \begin{equation}
        \frac{1}{T}\sum_{t=1}^{T}(Gu(t)-G_{k}u(t))^{2}=\frac{1}{T}\sum_{t=1}^{T}(\Delta G_{k} u(t))^{2},
    \end{equation}
    thus depending on the similarity between $G$ and $G_{k}$ (see \eqref{eq:plant_difference}). At the same time, these results indicate that (as expected) noise impacts on the average values of the proposed (normalized) indexes through its variance. In turn, this might make these indicators poor evaluators of similarities and observed performance when the noise is particularly high. Future work will thus be devoted to refine these indexes toward reducing the influence of noise on the regularization penalties.   
    \subsection{The overall data-driven problem}\label{sec:overall}
    Combining all the previous \textquotedblleft ingredients\textquotedblright, we can finally write the numerically tractable, data-driven counterpart of the ideal problem in \eqref{eq:ideal_stab_meta_problem2} as 
    \begin{subequations}\label{eq:DD_stab_meta_problem}
		\begin{align}
			&\underset{\alpha}{\mathrm{minimize}}~~J^{d,IV}(\alpha)\!+\!\lambda_{J}R_{J}(\alpha;\mathbf{\tilde{J}^{d}})\!+\!\lambda_{S}R_{S}(\alpha;\mathbf{S})\label{eq:DD_regularizedloss}\\
			& \qquad~~ \mbox{s.t. }~~C(q^{-1};\alpha)=\sum_{k=1}^{N}[\alpha]_{k}C(q^{-1};\theta_{k}),\\
			& \qquad~~\qquad~ \alpha_{k} \geq 0,~~k=1,\ldots,N,\\
			& \qquad~~\qquad~\sum_{k=1}^{N} \alpha_{k}=1,\\
			& \qquad~~\qquad~\hat{\delta}(\alpha) \leq \delta, \label{eq:DD_stabconstr}
		\end{align}
	\end{subequations}
    where $\mathbf{\tilde{J}^{d}}=\{\tilde{J}_{k}^{d}\}_{k=1}^{N}$ and $\mathbf{S}=\{S_{k}\}_{k=1}^{N}$.

    In what follows (including the numerical example), we will consider the following as a reasonable choice of the regularizers:
    \begin{subequations}\label{eq:FOC_reg}
    \begin{equation}
        R_{S}(\alpha;\mathbf{S})=\|S^{\top}\alpha\|_{1},~~~R_{J}(\alpha;\mathbf{\tilde{J}^{d}})=\|\alpha\|_{\mathrm{\tilde{J}^{d}}}^{2},
    \end{equation}
    with 
    \begin{equation}
    S^{\top}\!\!\!=\!\begin{bmatrix}
            S_{1}^{\top} & \ldots & S_{N}^{\top}
        \end{bmatrix}^{\!\top}\!\!,~~\mathrm{\tilde{J}^{d}}=\mathrm{diag}\left(\tilde{J}_{1}^{d},\ldots,\tilde{J}_{N}^{d}\right).
    \end{equation}
    \end{subequations}
    In other words, we use our insights on similarity to promote shrinkage within $\alpha$ in \eqref{eq:meta_controller} via the 1-norm regularization. This choice is coherent with the idea that controllers designed for systems that are more similar to $G$ are also more likely to be effective on the latter. At the same time, thanks to the performance-oriented Tikhonov regularizer, we steer the elements of $\alpha$ towards similar values if they are associated to controllers that have been proven to perform comparably, while shrinking them whenever the experienced performance are poor. We wish to stress that these are only two possible regularization options, which we will compare to other alternatives in future works.
    \section{Meta-FOC of a brushless DC motor}\label{sec:numerical_example}
    \begin{figure}[!tb]
        \centering
        \includegraphics[scale=.8,trim=0cm 0cm 44.9cm 71cm,clip]{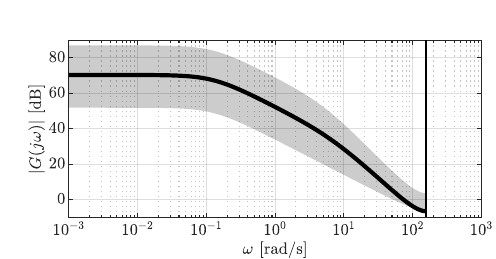}\vspace{-.2cm}
        \caption{Average magnitude of the frequency response of the family of DC motors (black line) and interval in which the other possible responses may lay (shaded area).}\label{fig:bode}
    \end{figure}
    \begin{figure}[!tb]
    \centering
    \includegraphics[scale=.8,trim=0cm 0cm 44.9cm 69.3cm,clip]{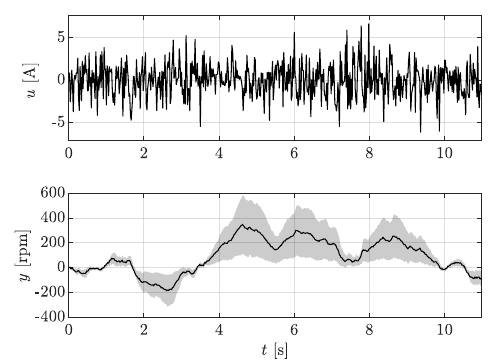}\vspace{-.2cm}
    \caption{Current/speeds pairs collected in open-loop from the $N$ meta-motors. The black line indicates their average output, while the shaded area spotlights the outputs' possible deviations.}
    \label{fig:dataset}
\end{figure}
    We now assess the impact of the proposed strategy on a problem of practical relevance, namely the calibration of the PI controller within a Field-Oriented Control (FOC) scheme for the regulation of a brushless DC motor, according to the scheme in \cite{busetto2023data}. The controller in charge of generating the quadrature axis current $u(t)$~[A] belongs to the following family:
    \begin{equation}\label{eq:DC_PI_contr}
        C(q^{-1};\theta)\!=\!K_{P}\!+\!K_{I}\frac{T_{s}}{2}\frac{1+q^{-1}}{1-q^{-1}},~~~~\theta\!=\!\begin{bmatrix}
            K_{P}\\K_{I}
            \end{bmatrix},
    \end{equation}
    where $T_{s}=0.02$~[s] denotes the sampling time. Our goal is to calibrate $\theta$ in \eqref{eq:DC_PI_contr} for the closed-loop, noiseless motor speed $y^{\mathrm{o}}(t)$ [rpm] to match the output of the ideal target behavior dictated by
    \begin{equation}\label{eq:DC_desired}
        M(q^{-1})=\frac{0.0609q^{-1}}{1-0.9391q^{-1}}.
    \end{equation}
 
    In our simulations, we generate data assuming the dynamical structure of the motor to be given by
    \begin{subequations}\label{eq:DC_family_plant}
    \begin{equation}
        G(q^{-1})=\frac{\kappa q^{-2}}{(1-p_{1}q^{-1})(1-p_{2}q^{-1})},
    \end{equation}
     with $p_{1}=0.9975$, $\kappa \in [1.00,5.75],$ and $p_{2} \in [0,0.9]$.
    The resulting family of motors has a frequency response with magnitude reported in \figurename{~\ref{fig:bode}}, and a level of similarity $\varepsilon=784.55$ (see again \eqref{eq:varepsilon_similarity}).
    \end{subequations}
    Both $G$, indicated from now on as the \emph{new motor}, and the $N=10$ ones used to construct the meta-dataset (that we will refer to as \emph{meta-motors}) are uniformly sampled at random from this family, \emph{i.e.,} the parameters characterizing their dynamics are extracted based on the following sampling rules:
    \begin{equation}
         \kappa \sim\mathcal{U}_{[1,5.75]},~~~~p_{2}\sim \mathcal{U}_{[0,0.9]}.
    \end{equation}    
    In our tests, we always consider $10$ different realizations of the \emph{new motor}, for the outcome of our analysis not to be linked to a specific realization of the system's parameters. We wish to remark that \eqref{eq:DC_family_plant}, along with the true parameters characterizing the dynamics of each motor, are assumed to be unknown and, thus, not exploited for design purposes.

    To construct the dataset, both the new motor and the meta-motors are excited in open loop with the same white, Gaussian distributed quadrature current, with zero mean and standard deviation of 2~[A]. The measured speeds gathered over 11~[s] of experiment are then corrupted by a white, zero-mean, Gaussian noise sequence with standard deviation 10~[rpm] (see \figurename{~\ref{fig:dataset}} for a snapshot of the data collected from the meta-motors), yielding an average signal-to-noise ratio of $\overline{\mathrm{SNR}}\simeq 21.01~\mathrm{[dB]}$. We wish to stress that the resulting sets of input/output data are relatively \emph{small}, since they comprise $T=550$ samples only (typical lengths are of the order of thousands). Closed-loop experiments are carried out with the same level of noise acting on the measured outputs (and thus fed back to the controller). 
    
    The data collected from the $N=10$ meta-motors are further used to design $10$ different controllers with the structure in \eqref{eq:DC_PI_contr}. To this end, here we employ the direct (model-reference) control strategy proposed in \cite{busetto2023data}\footnote{Since it leverages SMGO-$\Delta$ \cite{SABUG2022}, this tuning strategy requires one to select the maximum number of iterations, here set to $70$, and the value of $\Delta$, that we impose equal to 0.5. Moreover, we set the search spaces for the proportional and integral gains as $[10^{-3},1]$ and $[10^{-4},10^{-1}]$, respectively, to span over a large number of candidate parameters.}. The latter entails a closed-loop calibration experiment, that we have carried out for $3$~[s] by considering a step reference with amplitude $1000$~[rpm]. Note that, apart from stabilizing the plant each controller has been designed for all these $N=10$ data-driven controllers stabilize all the new motors we have extracted and tested. Hence, they are never discarded from the meta-dataset. In learning the meta-controller, we consider a unitary weighting filter $F(q^{-1})$ (see the ideal cost in \eqref{eq:ideal_meta_loss}), while we use the regularization terms introduced in Section~\ref{sec:overall}. 
    
    All results reported hereafter have been obtained by solving the optimization problem in \eqref{eq:DD_stab_meta_problem} with the CVX package \cite{gb08,cvx} on an Intel(R) Core(TM) i7-10875H CPU @ 2.30GHz processor with 16 GB of RAM running MATLAB R2021b.

    \subsection{The benefits of meta-design}
    \begin{figure}[!tb]
        \centering
        \begin{tabular}{c}
        \subfigure[Meta-controller\label{fig:comparison_metaControl}]{\includegraphics[scale=.8,trim=0cm 0cm 44.9cm 71cm,clip]{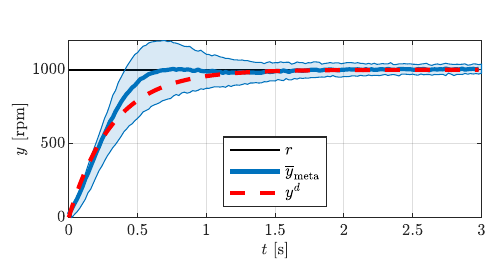}}\vspace{-.3cm}\\
        \subfigure[Controller tuned with SMGO-$\Delta$\label{fig:comparison_SMGO}]{\includegraphics[scale=.8,trim=0cm 0cm 44.9cm 71cm,clip]{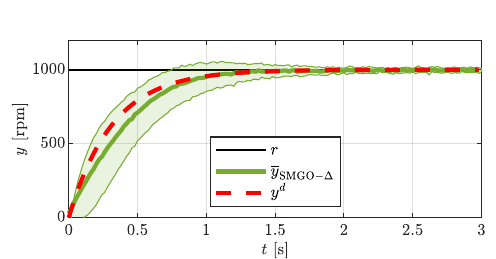}}\vspace{-.3cm}\\
        \subfigure[Controller tuned with the c-VRFT approach\label{fig:comparison_VRFT}]{\includegraphics[scale=.8,trim=0cm 0cm 44.9cm 71cm,clip]{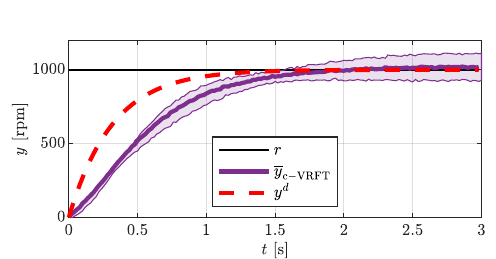}}
        \end{tabular}\vspace{-.2cm}
        \caption{Comparison of data-driven techniques: set point (black line) and desired response (dashed red line) \emph{vs} mean (colored line) and standard deviation (shaded area) of the closed-loop responses attained with different controllers.}\label{fig:comparison_methods}
    \end{figure}
    \begin{figure}[!tb]
        \centering
        \includegraphics[scale=.8,trim=0cm 0cm 44.9cm 67cm,clip]{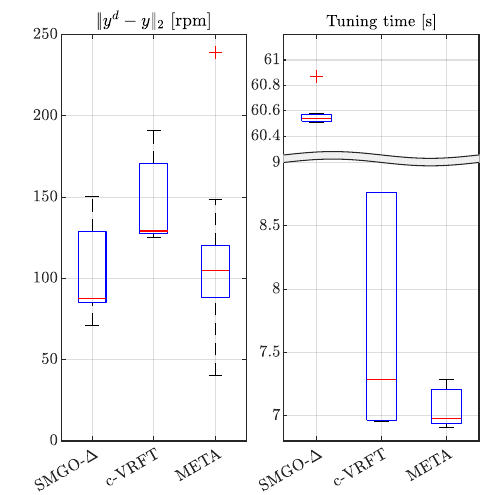}\vspace{-.2cm}
        \caption{Comparison of data-driven techniques: mismatching error (left panel) and tuning time (right panel). The time required for calibration includes that to carry out all the needed experiments, which are $2$ for both the meta-learning and the VRFT approach {with stability constraint} (c-VRFT), and equal to the number of iterations for SMGO-$\Delta$.}
        \label{fig:boxplot_comparison}
    \end{figure}
    \begin{figure}[!tb]
        \centering
        \includegraphics[scale=.8,trim=0cm 0cm 44.9cm 71cm,clip]{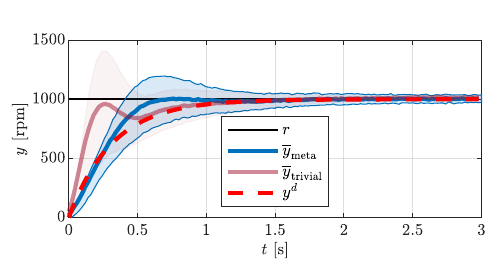}\vspace{-.2cm}
        \caption{Comparison with the \textquotedblleft trivial\textquotedblright \ meta-controller: set point (black line) and desired response (dashed red line) \emph{vs} mean (colored line) and standard deviation (shaded area) of the attained closed-loop outputs.}\label{fig:comparison_trivial}
    \end{figure}
    \begin{figure}[!tb]
        \centering
        \begin{tabular}{c}
            \subfigure[Normalized $S_{k}$ for each meta-motor]{\includegraphics[scale=.8,trim=0cm 0cm 44.9cm 71cm,clip]{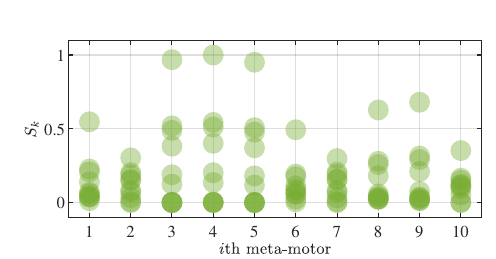}}\vspace{-.3cm}\\ 
            \subfigure[Normalized $\tilde{J}_{k}^{d}$ for each $C_{k}$ in $\mathcal{D}_{N}^{\mathrm{meta}}$]{\includegraphics[scale=.8,trim=0cm 0cm 44.9cm 71cm,clip]{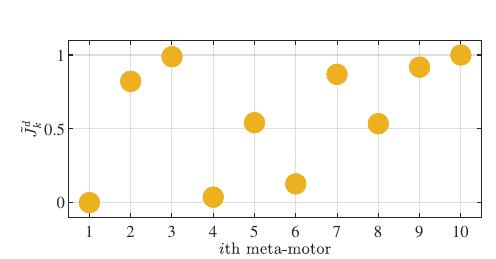}}\vspace{-.3cm}\\
            \subfigure[Elements of $\alpha$]{\includegraphics[scale=.8,trim=0cm 0cm 44.9cm 71cm,clip]{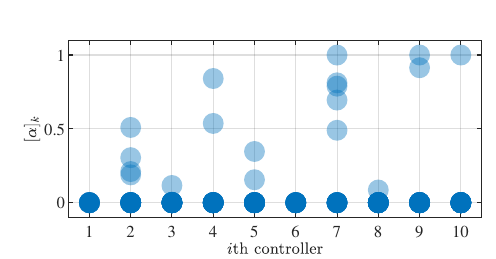}}
            \end{tabular}\vspace{-.2cm}
            \caption{Overview of the meta-dataset and outcomes of meta-control design over different experiments: darker circles for the similarity indexes and the values of $[\alpha]_k$ indicate the values that occur more frequently.}\label{fig:usage_vsperformance_similarity}
    \end{figure}
    We initially compare the performance attained by deploying the meta-controller calibrated by solving problem \eqref{eq:DD_stab_meta_problem} (with penalties fixed at $\lambda_{J}=30$ and $\lambda_{S}=300$) and two others, respectively tuned with the iterative method proposed in \cite{busetto2023data} and the VRFT approach \cite{campi2002}. \textcolor{black}{In solving \eqref{eq:DD_stab_meta_problem} and using the approach proposed by \cite{busetto2023data} we neglect the stability constraint (see \eqref{eq:DD_stabconstr} and \cite[Section 4]{van2011data}, respectively), as both methods always result in stable closed-loops. This is not the case with the VRFT approach that, without the stability constraint, results in unstable closed-loop behaviors in 50\% of the tested cases. We have thus decided to consider the VRFT scheme equipped with the stability constraint with $\delta=0.5$ for our comparison to be fairer (\emph{a.k.a.,} not biased by the unstable behaviors).} 

   \textcolor{black}{
    \begin{remark}[Practical choice $\ell$]\label{remark:choice_ell}
    A window of length $\ell=200$ is generally used when computing $\hat{\delta}(\alpha)$, as discussed in Section~\ref{sec:stability_datadriven}. Nonetheless, sometimes this choice leads to a failure of the employed solver \cite{SDP3,Tutuncu2003}. In this case, $\ell$ is decreased to $10$ and then successively increased up to the maximum value for which the solver is able to retrieve a solution for \eqref{eq:DD_stab_meta_problem}. Note that this operation increases the computational time needed to design the controller.
    \end{remark}}
    
    
    By looking at the responses reported in \figurename{~\ref{fig:comparison_methods}}, it is clear that the meta-controller and the one tuned with SMGO-$\Delta$ lead to similar closed-loop performance, allowing the closed-loop system to mimic the desired behavior. Even if the second controller enables the closed-loop system to attain an average response that is closer to the desired one, this slight improvement in the average matching comes at the price of a considerable increase in tuning time (see \figurename{~\ref{fig:boxplot_comparison}}). Instead, \textcolor{black}{the VRFT approach returns controllers that result in considerably poorer performance, both in tracking the desired response and the reference signal, most probably due to the limited dimension of the considered dataset $\mathcal{D}_{T}$.} 
    These results clearly show the benefits of the proposed meta-learning rationale, which allows for a trade-off between tuning time/effort and performance.

    As a final remark for this section, we wish to stress the importance of optimizing the weights of the meta-controller, more than simply using the information coming from the meta-dataset. Let us consider the \textquotedblleft trivial\textquotedblright \ meta-controller as the one with
    \begin{equation*}
        [\alpha]_{k}=\frac{1}{N},~~~k=1,\ldots,N.       
    \end{equation*}
    The comparison in \figurename{~\ref{fig:comparison_trivial}} shows that such a controller results in worst matching throughout the transient with respect to the optimized one. This result is somehow expected, since the proposed tuning strategy actually leverages all information available in $\mathcal{D}$ (see \eqref{eq:D}) to calibrate the controller, including insights on the similarities between plants and the performance experienced with the different controllers. In turn, this leads the optimizer to \textquotedblleft prefer\textquotedblright \ some controllers within the meta-dataset over others, with the least performing ones that are either never or rarely considered in the construction of the controller for the new motor (see \figurename{~\ref{fig:usage_vsperformance_similarity}}).
    
    \subsection{Assessing non-deteriorating performance}
    \begin{figure}[!tb]
    \centering
    \begin{tabular}{cc}
    \subfigure[$\lambda_S=300$]{\includegraphics[scale=.8,trim=0cm 0cm 50.25cm 69.25cm,clip]{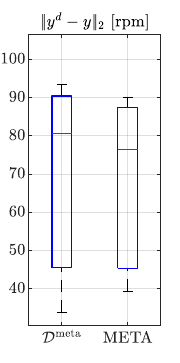}} & \subfigure[$\lambda_S=3000$]{\includegraphics[scale=.8,trim=0cm 0cm 50.25cm 69.25cm,clip]{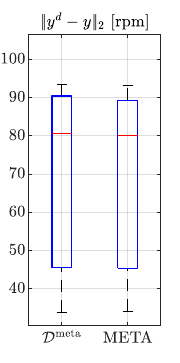}}\\
    \end{tabular}
    \caption{\textcolor{black}{Non-deteriorating performance with $\lambda_J=30$: 2-norm of the mismatching error in closed-loop.}}\label{fig:non-decreasing1}
    \end{figure}
        \begin{figure}[!tb]
    \centering
    \begin{tabular}{c}
    \subfigure[$\lambda_S=300$]{\includegraphics[scale=.8,trim=0cm 0cm 44.25cm 71.25cm,clip]{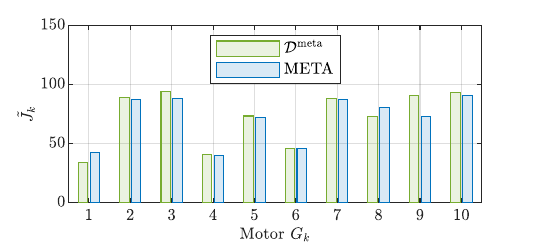}}\vspace{-.3cm}\\ 
    \subfigure[$\lambda_S=3000$]{\includegraphics[scale=.8,trim=0cm 0cm 44.25cm 71.25cm,clip]{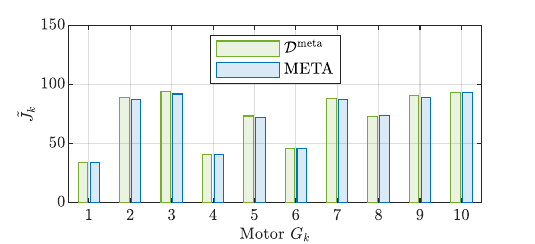}}\vspace{-.3cm}
    \end{tabular}
    \caption{Non-deteriorating performance with $\lambda_J=30$: 2-norm of the mismatching error in closed-loop for each of the 10 tests picking one of the motors in $\mathcal{D}^{\mathrm{meta}}$.}\label{fig:non-decreasing2}
    \end{figure}
    \begin{figure}[!tb]
    \centering
    \includegraphics[scale=.8,trim=0.1cm 0cm 47.8cm 69cm,clip]{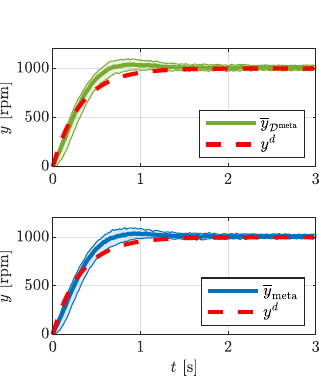}
    \caption{\textcolor{black}{Non-deteriorating performance with $\lambda_S=300$ and $\lambda_J=30$: and mean (line) and standard deviation (shaded area) of closed-loop responses contained in the meta-dataset [top panel] and attained with the meta-controller [low panel] \emph{vs} the desired output (dashed red line).}}\label{fig:non-decreasing3}
    \end{figure}
    We now empirically evaluate the capability of $C(\alpha)$ to result in non-deteriorating performance, to test the robustness of the property in Proposition~\ref{prop:non_decrease} in a non-idealized setting. To this end, we still consider $10$ different new motors, that are in turn equal to one of the meta-motors used to construct $\mathcal{D}_{N}^{\mathrm{meta}}$. 
    
    As shown in \textcolor{black}{\figurename{~\ref{fig:non-decreasing1}} and \figurename{~\ref{fig:non-decreasing2}}}, the proposed meta-control rationale \textcolor{black}{generally} allows us to attain improved \textcolor{black}{or, at least, similar} model-reference matching with respect to the one achieved by using the controllers in the meta-dataset \textcolor{black}{depending on the chosen trade-off between $\lambda_{S}$ and $\lambda_{J}$ in \eqref{eq:DD_regularizedloss}. Indeed, lower values of $\lambda_{S}$ tend to enhance the overall performance attained by using the meta-controller with respect to that achieved leveraging the controllers within $\mathcal{D}^{\mathrm{meta}}$ (as also confirmed by the closed-loop responses shown in \figurename{~\ref{fig:non-decreasing3}}). This improvement comes at the price of having two instances (namely, $G_{1}$ and $G_{8}$) for which the meta-controller results in a slight deterioration of performance. Nonetheless, by increasing $\lambda_{S}$ to $3000$, the performance of the meta-controller are indeed non-deteriorating, even if the difference in the closed-loop matching attained with the meta-controller and the ones within $\mathcal{D}^{\mathrm{meta}}$ becomes overall less relevant.}

    \textcolor{black}{These results thus show that the proposed approach can recognize the meta-motor equal to the new one (thus prioritizing the associated controller), while still generally benefiting from performing controllers tailored to meta-motors that are also similar to the new one as long as $\lambda_{S}$ does not excessively dominate over $\lambda_{J}$.}


    \subsection{Sensitivity analysis to the regularization penalties}
    \begin{figure}[!tb]
    \centering
    \includegraphics[width=\linewidth,trim=0cm 0cm 0cm .8cm,clip]{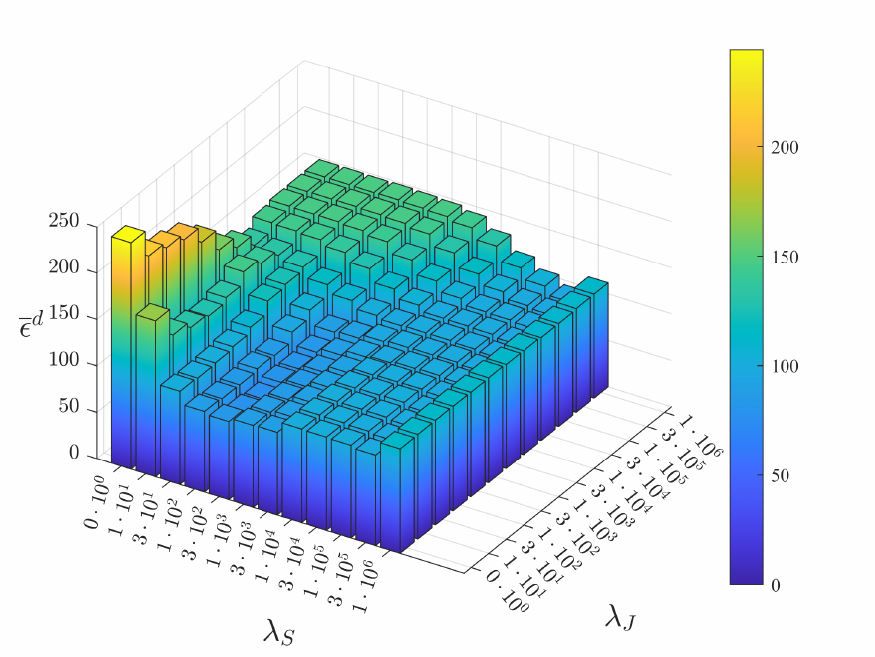}
    \caption{Sensitivity analysis: average 2-norm of 
    the closed-loop mismatching error $\overline{\epsilon}^{d}$ [rpm] over the $10$ realizations of the new motor \emph{vs} $\lambda_S$ and $\lambda_L$. The values of $\overline{\epsilon}^{d}$ above $250$ have been saturated to this value for visualization purposes.}
    \label{fig:sensitivity}
\end{figure}
    Going back to the scenario where the $N=10$ new motors are sampled at random, we now evaluate the sensitivity of the attained closed-loop response to the choice of $\lambda_{S}$ and $\lambda_{J}$ in \eqref{eq:DD_stab_meta_problem}, while still neglecting the stability constraint. The results of our analysis are reported in \figurename{~\ref{fig:sensitivity}}, clearly showing that the meta-control design procedure tends to result in poorer performance when both regularization penalties are low, especially for what concerns $\lambda_{S}$. This result is somehow expected since low values of $\lambda_{S}$ do not leverage insights on the plant similarities. At the same time, also high values of $\lambda_{S}$ result in a drop in performance. In this case, \eqref{eq:DD_stab_meta_problem} tends to prioritize controllers based on their similarity only, even when the performance already experienced with them is poor. In turn, this potentially hampers the achievement of the desired closed-loop behavior. \figurename{~\ref{fig:sensitivity}} further highlights the importance of retaining lower values of $\lambda_{J}$, due to the fact that higher $\lambda_{J}$ would result in prioritizing performance over similarity. This might be undesired, especially when the most performing controllers in $\mathcal{D}_{N}^{\mathrm{meta}}$ are the ones associated with meta-motors that are more different from the new one. Note that even when $\lambda_{S}=\lambda_{L}=0$ the meta-controller attains an average 2-norm of the matching error in closed-loop of about $244.54$ [rpm], which is comparable to the one achieved by designing a brand new PI with the VRFT approach, that is equal to $241.35$ [rpm].

    \subsection{On the impact of the meta-dataset's features}
    \begin{figure}[!tb]
        \centering
        \centering
        \includegraphics[scale=.8,trim=0cm 0cm 44.9cm 69.3cm,clip]{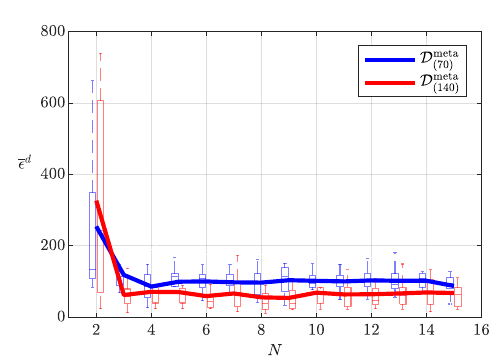}\vspace{-.2cm}
        \caption{Average 2-norm of the closed-loop mismatching error $\overline{\epsilon}^{d}$ [rpm] \emph{vs} dimension of the meta-dataset and number of SMGO-$\Delta$ iterations. $\mathcal{D}_{(70)}^{\mathrm{meta}}$ and $\mathcal{D}_{(140)}^{\mathrm{meta}}$ denote the results obtained by running SMGO-$\Delta$ for $70$ and $140$ iterations, respectively.}\label{fig:boxplot_dimension}
    \end{figure}
    We now focus on how the performance attained with the meta-controller is affected by the size $N$ of the meta-dataset and the \textquotedblleft quality\textquotedblright \ of the included controllers. To this end, we consider $14$ meta-datasets of increasing dimension\footnote{The meta-dataset of size $N=10$ corresponds to the one used in all other analysis, while $\mathcal{D}_{N-1}^{\mathrm{meta}} \subset \mathcal{D}_{N}^{\mathrm{meta}}$.}, and two possible values for the maximum number of iterations of SMGO-$\Delta$, namely $70$ and $140$. As expected, \figurename{~\ref{fig:boxplot_dimension}} highlights that closed-loop matching generally improves when more iterations of SMGO-$\Delta$ are performed, concurrently with a decrease of the average value of the index $\tilde{J}_{k}^{d}$\footnote{As an example, for $N=10$ the average $\tilde{J}_{k}^{d}$ is equal to $5.51 \cdot 10^{3}$ when $70$ iterations are performed, while it drops to $4.13 \cdot 10^{3}$ when $140$ SMGO-$\Delta$ iterations are carried out.}. Only the results obtained for $N=2$ are not consistent with this trend, because the limited dimension of the meta-dataset hampers the overall closed-loop behavior and, thus, the matching performance. Meanwhile, performance tends to improve up to $N=4$, then reaching a plateau (on average). This outcome indicates that, at least for the example at hand, a meta-dataset of dimension $N=4$ would be sufficiently informative to design a performing meta-controller.  

    \subsection{The effect of the stability constraint}
    \begin{figure}[!tb]
        \centering
        \begin{tabular}{c}
             \subfigure[Matching performance and tuning time\label{fig:stab_boxplot}]{\includegraphics[scale=.8,trim=0cm 0cm 44.9cm 71cm,clip]{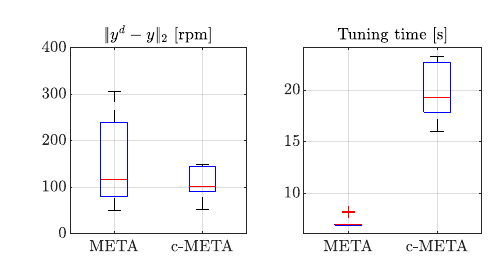}} \vspace{-.3cm}\\
              \subfigure[attained closed-loop responses]{\includegraphics[scale=.8,trim=0cm 0cm 44.9cm 69.3cm,clip]{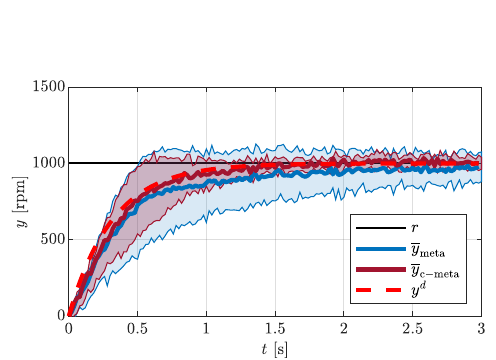}}
        \end{tabular}\vspace{-.2cm}
        \caption{Meta-control design with stability guarantees: set point (black line) and desired behavior (red dashed line) \emph{vs} mean (colored lines) and standard deviation (shaded areas) of the responses obtained with the meta-controllers designed with and without the stability constraint. The former is denoted as c-META.}\label{fig:incorporating_stability}
    \end{figure}
    We finally incorporate the stability constraint into the design problem, solving \eqref{eq:DD_stab_meta_problem} for $\delta=0.5$\textcolor{black}{, as employed in c-VRFT. The window length $\ell$ needed as for Section~\ref{sec:stability_datadriven} is once again chosen as detailed in Remark~\ref{remark:choice_ell}.}

    Since the introduction of the stability constraint has not resulted in changes in performance (but only in an increased design time) with the level of noise considered in the previous analysis, we now increase the standard deviation of the noise corrupting the outputs in $\mathcal{D}_{T}$ to $40$ [rpm] ($\overline{\mathrm{SNR}}=10.61$ [dB]), while keeping the meta-dataset of size $N=10$ unchanged with respect to the one exploited in the other analysis. The results we obtained are reported in \figurename{~\ref{fig:incorporating_stability}}. Clearly, they show that the introduction of the stability constraint in this more challenging setting can be of help in effectively coping with the noise acting on the data. Indeed, we observe a significant reduction in instances where matching performance get further away from the average. This benefit comes at the price of an increased design time (see \figurename{~\ref{fig:stab_boxplot}}), whose considerable variance is linked to the procedure exploited to adjust the window length $\ell$ every time the solver fails.       
     
    \section{Conclusions}\label{sec:conclusions}
    In this work, we have employed for the first time a \textit{meta-learning rationale} to enhance both the effectiveness and efficiency of direct, data-driven model reference control design. Like humans gain knowledge from past experiences, we propose to leverage controllers already calibrated for similar systems and data-based insights on their experienced performance and similarities to formulate a novel, meta-design problem. The numerical study carried out to calibrate a PI controller for a brushless DC motor has highlighted the potential of the method to enhance closed-loop performance when a reduced amount of data is collected, while at the same time resulting in a reduced tuning time.

     Future work will be devoted to setting the ground for a theoretical framework to analyze the informativity of the meta-dataset. This development would be crucial to extend the proposed approach to other (possibly iterative) data-based control techniques, since it would allow one to limit the dimension of the meta-dataset, by incorporating only the most informative incoming data in it. 

    \bibliographystyle{abbrv}
    \bibliography{meta_control.bib}
\end{document}